\documentclass[runningheads]{llncs}
\usepackage{amsmath, amssymb}
\usepackage{todonotes}
\usepackage{url}
\let\doendproof\endproof
\renewcommand\endproof{~\hfill$\qed$\doendproof}
\newcommand{\Runs}{\mathit{runs}}
\newcommand{\PS}{\mathit{psq}}

\begin{document}
\title{Lyndon Words, the Three Squares Lemma, and Primitive Squares}
\author{
  Hideo~Bannai\inst{1}\orcidID{0000-0002-6856-5185}\and
  Takuya~Mieno\inst{2,3}\orcidID{0000-0003-2922-9434} \and
  Yuto~Nakashima\inst{2}\orcidID{0000-0001-6269-9353}
}

\authorrunning{H. Bannai et al.}
\institute{
  M\&D Data Science Center, Tokyo Medical and Dental University, Tokyo, Japan\\
\email{hdbn.dsc@tmd.ac.jp}  \and
  Department of Informatics, Kyushu University, Fukuoka, Japan\\
\email{\{takuya.mieno,yuto.nakashima\}@inf.kyushu-u.ac.jp} \and
  {Japan Society for the Promotion of Science, Tokyo, Japan}
}
\maketitle

\begin{abstract}
  We revisit the so-called ``Three Squares Lemma''
  by Crochemore and Rytter [Algorithmica 1995]
  and, using arguments based on Lyndon words,
  derive a more general variant
  which considers three overlapping
  squares which do not necessarily share a common prefix.
  We also give an improved upper bound of $n\log_2 n$ on the
  maximum number of (occurrences of) primitively
  rooted squares in a string of length $n$,
  also using arguments based on Lyndon words.
  To the best of our knowledge, the only known upper bound
  was $n \log_\phi n \approx 1.441n\log_2 n$,
  where $\phi$ is the golden ratio,
  reported by Fraenkel and Simpson [TCS 1999]
obtained via the Three Squares Lemma.
\end{abstract}

\section{Introduction}
Periodic structures of strings have been and still are one of
the most important and fundamental objects
of study in the field of combinatorics on words~\cite{Berstel:2007fb},
and the analysis and exploitation of their combinatorial properties
are a key ingredient in
the development of efficient string processing algorithms~\cite{DBLP:journals/siamcomp/KnuthMP77,DBLP:conf/stoc/KempaK19}.

In this paper, we focus on {\em squares}, which are strings of the form $u^2~(= uu)$ for some string $u$, which is called the {\em root} of the square.
A well known open problem concerning squares is on the maximum number of distinct squares that can be contained in a string.
Fraenkel and Simpson~\cite{Fraenkel:1998wi} showed that the maximum number of distinct square substrings of a string of length $n$ is at most $2n$.
Although slightly better upper bounds of $2n-\Theta(\log n)$~\cite{Ilie:2007hv} and $\frac{11}{6}n$~\cite{Deza:2015ih} have been shown,
it is conjectured that it is at most $n$~\cite{Fraenkel:1998wi},
with a best known lower bound of $n - o(n)$~\cite{Fraenkel:1998wi}.

The ``Three Squares Lemma'' by Crochemore and Rytter~\cite{crochemore1995squares}
was the key lemma used by Fraenkel and Simpson to obtain the upper bound of $2n$.
\begin{lemma}[Three Squares Lemma~(Lemma 10 of~\cite{crochemore1995squares}\footnote{In~\cite{crochemore1995squares},
      $u, v, w$ are all assumed to be primitive and
      $|u| > |v|+|w|$ was claimed, but it was noted in~\cite{Fraenkel:1998wi} that
      only primitivity of $w$ is required, and that
      $|u| \geq |v|+ |w|$ is the correct relation,
      giving $u=01001010$, $v=01001$, and $w=010$ as an example when $|u| = |v|+|w|$.
    })]\label{lem:tslorig}
  Let $u^2$, $v^2$, $w^2$ be three prefixes of some string
  such that $w$ is primitive and $|u| > |v| > |w|$.
  Then, $|u| \geq |v| + |w|$.
\end{lemma}
Crochemore and Rytter
further showed
that the lemma implies that
the number of primitively rooted squares that can
start at any given position of a string is bounded by
$\log_\phi |x|$, where $\phi = (1+\sqrt{5})/2$ is the golden ratio (Theorem 11 of~\cite{crochemore1995squares}).
Thus, it follows that
the maximum number $\PS(n)$ of occurrences of primitively rooted squares in a string of length $n$
is less than $n\log_\phi n \approx 1.441n\log_2 n$.

The original proof of the Three Squares Lemma
by Crochemore and Rytter was based on the
well known ``Periodicity Lemma'' by Fine and Wilf~\cite{periodicity_lemma}.
Concerning a similar problem on
the maximum number of
``runs'' (maximally periodic substring occurrences such that
the smallest period is at most half its length) that can be
contained in a string,
the Periodicity Lemma was also the tool of choice in its analysis~\cite{814634,Rytter2006STACS,PUGLISI2008165,CROCHEMORE2008796}.
However,
this changed when
Bannai et al.~\cite{DBLP:conf/soda/BannaiIINTT15,Bannai:2017gv}
applied arguments
based on {\em Lyndon words}~\cite{Lyndon:1954it} to
solve, by a very simple proof,
a longstanding conjecture
that the maximum number of runs in a string of length $n$ is at most $n$.
Using the same technique,
the upper bound on the number of runs was further improved to $0.957n$
for binary strings~\cite{DBLP:conf/spire/0001HIL15}.
Bannai et al. also showed a new algorithm for computing all runs in a string,
which paved the way for algorithms with improved time complexity for general ordered alphabets
to be developed~\cite{KOSOLOBOV2016241,DBLP:conf/cpm/GawrychowskiKRW16,10.1007/978-3-319-46049-9_3}.

In this paper, we take the first steps of investigating to what extent Lyndon words can be
applied in the analysis of squares.
We first give an alternate proof of the Three Squares Lemma by arguments based
on Lyndon words, and extend it to show a more general variant which considers
three overlapping squares which do not necessarily share a common prefix.
Furthermore, we show a significantly
improved upper bound of $n\log_2 n$ on the maximum number of
occurrences of primitively rooted squares.

\section{Preliminaries}
Let $\Sigma$ be an alphabet.
An element of $\Sigma$ is called a symbol.
An element of $\Sigma^{\ast}$ is called a string.
The length of a string $w$ is denoted by $|w|$.
The empty string $\varepsilon$ is the string of length 0.
For any possibly empty strings $x,y,z$, if $w = xyz$, then
$x,y, z$ are respectively called a prefix, substring, suffix of $w$.
They are a {\em proper} prefix, substring, or suffix if they are not equal to $w$.
For any $1 \le i \le j \le |w|$, $w[i.. j]$ denotes the substring of $w$
starting at position $i$ and ending at position $j$.
We assume that $w[0],w[|w|+1]\neq w[i]$ for any $1 \leq i \leq |w|$.
For any string $x$, let $x^1 = x$, and for any integer $k \geq 2$, let $x^k = x^{k-1}x$.
If there exists no string $x$ and integer $k \geq 2$ such that $w = x^k$,
$w$ is said to be {\em primitive}.

A non-empty string $w$ is said to be a {\em Lyndon word}~\cite{Lyndon:1954it}
if $w$ is lexicographically smaller than any of its non-empty proper suffixes.
An important property of Lyndon words is that they cannot have a {\em border},
i.e., a non-empty substring that is both a proper suffix and prefix.
Also, notice that whether a string is a Lyndon word or not depends on
the choice of the lexicographic order.
Unless otherwise stated, our results hold for any lexicographic order.
However, we will sometimes require a pair of lexicographic orders $<_0$ and $<_{1}$,
the former induced by an arbitrary total order on $\Sigma$,
and the other induced by the opposite total order, i.e.,
for any $a,b \in \Sigma$, $a <_0 b$ if and only if $b <_1 a$.

An integer $1 \leq p \leq |w|$ is a {\em period} of string $w$ if
$w[i] = w[i+p]$ for all $i = 1,\ldots,|w|-p$.
A string is a {\em repetition}
if its smallest period $p$ is at most half of its length.
An occurrence $w[i..j]=v$ of a repetition $v$ with smallest period $p$
is a {\em maximal repetition} (or a {\em run}) in $w$,
if the smallest periods of both $w[i-1..j]$ and $w[i..j+1]$ are not $p$.

For any repetition $v$, an {\em L-root}~\cite{DBLP:journals/tcs/CrochemoreIKRRW14} $\lambda_v$
is a substring of $v$ that is a Lyndon word whose
length is equal to the smallest period of $v$.
It is easy to see that an L-root of a repetition always exists
and is unique.
We also define the {\em L-root interval} $r_v$
in $v$ as
the substring corresponding to the maximal integer power
in $v$ of $\lambda_v$.
Any repetition $v$ can be written as $v = xr_v y$ where
$x$ (resp. $y$) is a possibly empty proper suffix (resp. prefix) of
the L-root $\lambda_v$.
Notice that for any square $u^2$, $|r_{u^2}| \geq |u|$.
Also, for any square $u^2$,
it can be shown that
the smallest period $p_u$ of $u^2$ is a divisor of $|u|$
and is equal to $|\lambda_{u^2}|$, which implies that
it is also the smallest period of $r_{u^2}$ and
a divisor of $|r_{u^2}|$.

The next lemma shows that a Lyndon word can only occur
in a run as a substring of the L-root of the run.
\begin{lemma}
  \label{lem:substringOfLyndonPower}
  For any Lyndon word $v$, there is no Lyndon word
  $w = xyz$ for strings $x,y,z$ such that
  $x$ (resp. $z$) is a non-empty suffix (resp. prefix) of $v$.
\end{lemma}
\begin{proof}
  If such $w$ exists,
  $v \leq x < xyz = w < z \leq v$, a contradiction.
\end{proof}

\section{Squares and L-roots}
We first prove a lemma concerning two squares.
\begin{lemma}\label{lem:doublesquareroots}
  Let $u^2$ and $v^2$ be squares where $v^2$ is a proper prefix of $u^2$.
  Then, the L-root interval $r_{u^2}$ of $u^2$ is not a substring of $v^2$,
  and either $r_{v^2}$ is a prefix of $r_{u^2}$, or
$r_{v^2}$ ends before $r_{u^2}$ starts.
\end{lemma}
\begin{proof}
  Let $p_u$ and $p_v$ respectively be the smallest periods of $u^2$ and $v^2$.
If $r_{u^2}$ is a substring of $v^2$,
  then,
  $v^2 = x r_{u^2} y = w r_{v^2} z$
  for
  some suffix $x$ of $\lambda_{u^2}$,
  some prefix $y$ of $\lambda_{u^2}$,
  some suffix $w$ of $\lambda_{v^2}$,
  and some prefix $z$ of $\lambda_{v^2}$.
  If $p_u \neq p_v$, then $r_{u^2} \neq r_{v^2}$ must hold
  since $p_u$ and $p_v$ are respectively their smallest periods.
  This implies either $|x|\neq|w|$ or $|y|\neq|z|$.
  However, that would contradict
  Lemma~\ref{lem:substringOfLyndonPower}.
  If $p_u = p_v$, then it must be that $r_{u^2}=r_{v^2}$
  due to their maximality.
  Since $u$ is longer than $v$, and
  $p_u=p_v$ must also be a divisor of their lengths,
  $u^2$ must be at least $2p_u$ longer than $v^2$.
  However that would contradict the maximality of $r_{u^2}$,
  since at least one more copy of $\lambda_{u^2}$ would fit inside $u^2$.

  Next, suppose that $r_{v^2}$ overlaps with $r_{u^2}$,
  and is not a prefix of $r_{u^2}$.
  Since $r_{u^2}$ cannot be a substring of $v^2$ in which $r_{v^2}$ is a substring,
  $r_{u^2}$ starts in $v^2$, and ends after the end of $v^2$.
  There are two cases:
  (1)
$r_{v^2}$ starts after the beginning of $r_{u^2}$ and ends in $r_{u^2}$ (Fig.~\ref{fig:case1})
or (2) $r_{v^2}$ starts before $r_{u^2}$, and ends in $r_{u^2}$ (Fig.~\ref{fig:case2}).

  \begin{figure}[h]
    \begin{minipage}[c]{0.5\textwidth}
      \centerline{\includegraphics[width=0.95\textwidth]{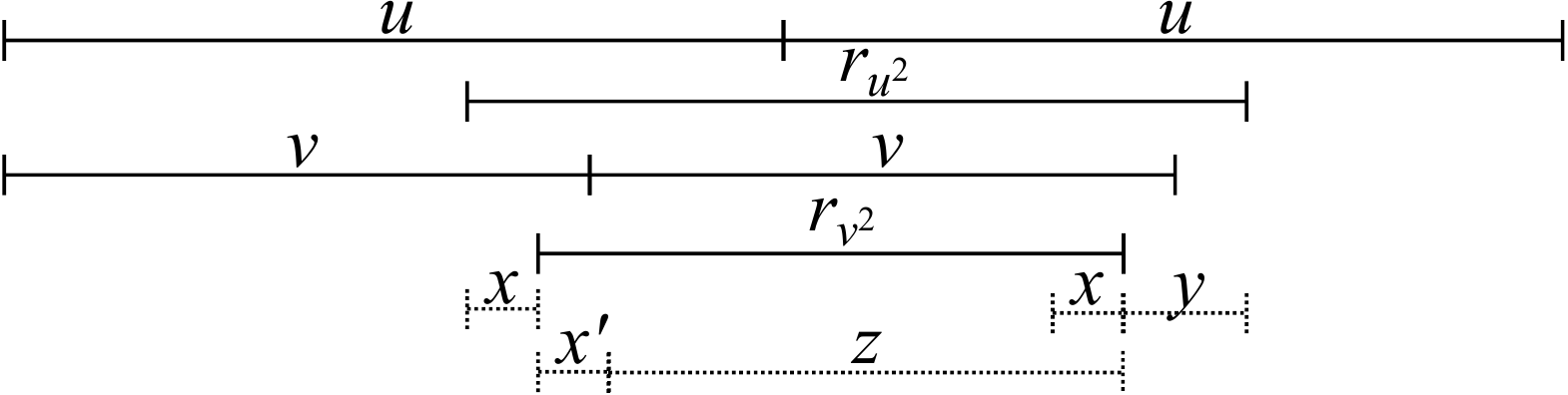}}
      \caption{Case (1) of Lemma~\ref{lem:doublesquareroots}.}
      \label{fig:case1}
    \end{minipage}
    \begin{minipage}[c]{0.5\textwidth}
      \centerline{\includegraphics[width=0.95\textwidth]{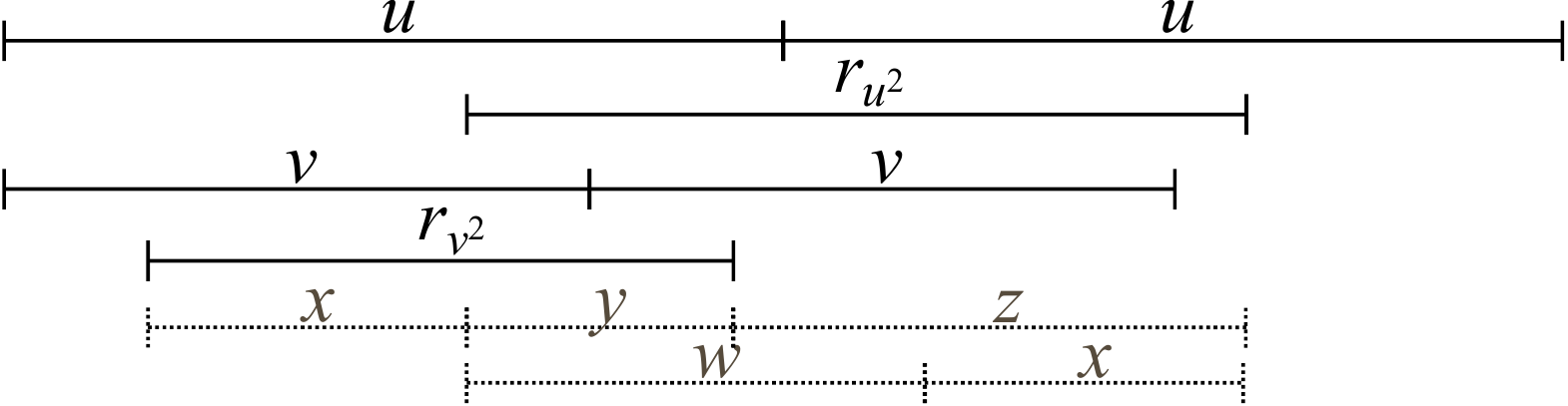}}
      \caption{Case (2) of Lemma~\ref{lem:doublesquareroots}.}
      \label{fig:case2}
    \end{minipage}
  \end{figure}
  Case (1) implies that $r_{u^2} = x r_{v^2} y$ for some
  non-empty proper suffix $x$ of $\lambda_{v^2}$
  and some suffix $y$ of $r_{u^2}$.
Let $r_{v^2} = x'z$ where $|x'|=|x|$.
  Since $|x'| < p_v$,
  we have $x > r_{v^2} > x'$,
  and thus, $r_{u^2} = x r_{v^2} y > x'zy$.
  This can hold only if
  $|x|$ is a multiple of $p_u$,
  but this also implies $x=x'$ which is a contradiction.

  Case (2) implies that a suffix of $r_{v^2}$ overlaps with a prefix of $r_{u^2}$.
  Let $r_{v^2} = xy$, $r_{u^2} = yz$ where $y$ is the overlap,
  and observe that $|x| < p_u$ due to the maximality of $r_{u^2}$.
  Notice that since $u^2$ has period $p_u$ which is a divisor of $|r_{u^2}|$,
  $x$ must also be a suffix of $r_{u^2}$,
  so we can write $r_{u^2} = wx$ for some $w$.
From Lemma~\ref{lem:substringOfLyndonPower},
$x$ must be an integer power of $\lambda_{v^2}$, since otherwise,
  there would be an occurrence of $\lambda_{v^2}$
  crossing the boundary of $x$ and $y$.
  Thus, $r_{u^2}$ contains the Lyndon word $\lambda_{v^2}$ of length $p_v$ as a prefix and suffix,
  which can only hold if $p_u = p_v$.
  However, this contradicts the maximality of $r_{u^2}$.
\end{proof}

To prove Lemma~\ref{lem:tslorig}, we use the previous lemma,
together with the following lemma used in the proof of the
``runs'' theorem~\cite{Bannai:2017gv}
which connects L-roots of runs
and longest Lyndon words starting at each position.
\begin{lemma}[Lemma 3.3 of~\cite{Bannai:2017gv}]
  \label{lem:longestLyndonWordAsLRoot}
  For any run $w[i..j]$ with period $p$,
  consider the lexicographic order
  $< \in \{ <_0, <_1 \}$ such that $w[j+1] < w[j+1-p]$.
  Then,
  any occurrence of the L-root of the run $w[i..j]$
  is the longest Lyndon word starting at that position.
\end{lemma}
It is easy to see that for any repetition,
there is a unique run with the same smallest period and L-root
in which the repetition is contained.
For any occurrence of a repetition in a string,
we will refer to the lexicographic order considered
in Lemma~\ref{lem:longestLyndonWordAsLRoot}
as {\em the} lexicographic order of the repetition.

\begin{proof}[of Lemma~\ref{lem:tslorig}]
  Consider the lexicographic order of $w^2$,
  i.e., L-root $\lambda_{w^2}$ is
  a longest Lyndon word starting at the first position of $r_{w^2}$.
  From Lemma~\ref{lem:doublesquareroots},
  the starting positions $b_{w^2},b_{v^2},b_{u^2}$ respectively
  of $r_{w^2}$, $r_{v^2}$, $r_{u^2}$ are non-decreasing.
  There are four cases:
  (1) $b_{w^2} < b_{v^2} < b_{u^2}$,
  (2) $b_{w^2} < b_{v^2} = b_{u^2}$,
  (3) $b_{w^2} = b_{v^2} < b_{u^2}$,
  and
  (4) $b_{w^2} = b_{v^2} = b_{u^2}$,
  where inequality of the starting positions
  implies the disjointness of the L-root intervals.

  Case (1): It follows that $r_{w^2}$, $r_{v^2}$, $r_{u^2}$ occur disjointly in $u^2$.
  Therefore, $2|u| \geq |r_{w^2}| + |r_{v^2}| + |r_{u^2}|$. Since $|r_{w^2}|\geq|w|, |r_{v^2}|\geq |v|, |r_{u^2}| \geq |u|$, we have
  $|u| \geq |w|+|v|$.

  Case (2): It follows that $r_{w^2}$ occurs disjointly before
  $r_{u^2}$, and $r_{v^2}$ is a prefix of $r_{u^2}$.
  Since $r_{v^2}\geq|v|$,
  $r_{w^2}$ is a substring of $v$ and thus also of $u$.
  Due to $u^2$ and $v^2$, there are two other
  occurrences of $r_{w^2}$ respectively $|u|$ and $|v|$ positions to the right.
Since $w$ is primitive, the smallest period of $r_{w^2}$ is $|\lambda_{w^2}|=|w|$,
  and thus the two occurrences of $r_{w^2}$ must be at least $|w|$ apart.
Therefore, $|w| \leq |u| - |v|$, which implies $|u| \geq |v|+|w|$.

Case (3):
  By the assumption of the lexicographic order,
  $\lambda_{w^2}$ is the longest Lyndon word starting at $b_{w^2}$
  and thus $|\lambda_{w^2}| \geq |\lambda_{v^2}|$.
  Since $r_{w^2}$ is a prefix of $r_{v^2}$,
  it must hold that $\lambda_{w^2} = \lambda_{v^2}$
  due to Lemma~\ref{lem:substringOfLyndonPower}.
  Since
  $|v| > |w| = |\lambda_{w^2}| = |\lambda_{v^2}|$
  and $|v|$ is a multiple of $|\lambda_{v^2}|$,
  we have $|v| \geq 2|\lambda_{v^2}|$.
  This implies $|r_{v^2}| \geq |v|+|\lambda_{v^2}|$.
  Also, since $r_{u^2}$ occurs disjointly with $r_{v^2}$ in $u^2$,
  we have
  $2|u| \geq |r_{v^2}|+|r_{u^2}|$, which implies
  $|u| \geq |r_{v^2}|$ since $|r_{u^2}| \geq |u|$.
  Then, $|u|\geq |r_{v^2}| \geq |v|+|\lambda_{v^2}| = |v|+|w|$.

  Case (4):
  Analogously to the previous case,
  we have $\lambda_{w^2} = \lambda_{v^2} = \lambda_{u^2}$.
  This implies that
  $|u|,|v|$ are multiples of $|\lambda_{w^2}|$
  and since $|u|>|v|$, we have
  $|u| \geq |v|+|\lambda_{v^2}| = |v|+|w|$.
\end{proof}

We note that actually, the proof of Lemma~\ref{lem:doublesquareroots}
does not require $v^2$ to be a prefix of $u^2$,
but only that $v^2$ is a substring of $u^2$ that starts before $r_{u^2}$,
so slightly stronger statements hold.
\begin{corollary}
  Let $u^2$ and $v^2$ be squares such that $v^2$ is a proper substring of
  $u^2$ that starts before the L-root interval $r_{u^2}$ of $u^2$.
  Then, $r_{u^2}$ is not a substring of $v^2$,
  and either
  the L-root interval $r_{v^2}$ of $v^2$
  is a prefix of $r_{u^2}$, or
$r_{v^2}$ ends before $r_{u^2}$ starts.
\end{corollary}
\begin{corollary}
  Let $u^2$, $v^2$, and $w^2$ be squares such that
  $v^2$ is a proper substring of $u^2$ that starts before $r_{u^2}$,
  and
  $w^2$ is a proper substring of $v^2$ that starts before $r_{u^2}$ and $r_{v^2}$,
  where
  $r_{u^2}$, $r_{v^2}$ are respectively the L-root intervals of $u^2,v^2$
  with respect to the lexicographic order of $w$.
  If $w$ is primitive, then $|u|\geq|v|+|w|$.
\end{corollary}
\section{Tighter upper bound for $\PS(n)$}
There can be $\Theta(n^2)$
occurrences of non-primitively rooted squares
in a string of length $n$
(e.g. a unary string).
However, as mentioned in the introduction,
Lemma~\ref{lem:tslorig} implies an upper bound of
$n\log_{\phi} n \simeq 1.441 n \log_2 n$
for $\PS(n)$, i.e.,
the maximum number of occurrences of primitively rooted squares in a string of length $n$.
On the other hand,
the best known lower bound is given by Fibonacci words,
which contain
$\frac{2(3-\phi)}{5\log_2\phi} F_n \log_2 F_n + O(F_n)$
occurrences of primitive squares~\cite{fraenkel1999exact},
where $F_n$ is the length of the $n$-th Fibonacci word,
$\phi$ is the golden ratio, and $\frac{2(3-\phi)}{5\log_2\phi} \approx 0.7962$.
Below, we prove a significantly improved upper bound for $\PS(n)$.
\begin{theorem} \label{thm:main}
  $\PS(n) \le n \log_2 n$.
\end{theorem}

Each primitively rooted square of $w$ is a substring of a run of $w$.
Let $\Runs(w)$ denote the set of runs in $w$.
Conversely, each run $\rho \in \Runs(w)$ with length $\ell_{\rho}$ and period $p_{\rho}$ contains exactly
$\ell_{\rho}-2p_{\rho}+1$ primitively rooted squares as substrings.
Let $\lambda_{\rho}$ be an L-root of a run $\rho$ with respect to
the lexicographical order of $\rho$.
If we consider the rightmost occurrence of $\lambda_{\rho}$ in $\rho$,
there exist strings $x_{\rho}, y_{\rho}$ such that
$\rho = x_{\rho} \lambda_{\rho} y_{\rho}$ and $y_{\rho}$ is a possibly empty proper prefix of $\lambda_{\rho}$.
Since $|\lambda_{\rho}| \geq |y_{\rho}|+1$,
the number of primitively rooted squares in
$\rho$ is $\ell_{\rho} - 2p_{\rho}+1 = |x_{\rho}|+|\lambda_{\rho}|+|y_{\rho}| - 2|\lambda_{\rho}|+1\leq |x_{\rho}|$.
Thus, the total sum of $|x_{\rho}|$ for all runs in $w$ gives an upper bound on
the number of occurrences of primitively rooted squares in $w$.
We will show that this total sum
is bounded by $n\log_2 n$ for any string $w$ of length $n$, which will yield Theorem~\ref{thm:main}.

To this end, we use the notion of Lyndon trees~\cite{BARCELO199093,Bannai:2017gv}.
The {\em Lyndon tree} of a Lyndon word $w$ is an ordered full binary tree defined recursively as follows\footnote{If $w$ is not a Lyndon word, we simply consider
the Lyndon word obtained by prepending to $w$ a symbol smaller than any symbol in $w$.}:
If $|w|=1$, then the Lyndon tree of $w$ is a single node labeled $w$, and if $|w|\geq 2$, then the root is labeled $w$, and the left and right children of $w$ are respectively the Lyndon trees of $u$ and $v$,
where $w = uv$ and $v$ is the lexicographically smallest proper suffix of $w$.
Note that this is known as the standard factorization of $w$~\cite{10.2307_1970044,lothaireCombinatoricsOnWords}, and $u,v$ are guaranteed to be Lyndon words.

From Lemma~\ref{lem:longestLyndonWordAsLRoot} and Lemma~\ref{lem:longestIsRightnode} below,
we have that for any string $w$,
the right nodes of the two Lyndon trees of $w$ with respect to $<_0$ and $<_1$
contain all L-roots of all runs in $w$.
\begin{lemma}[Lemma 5.4 of~\cite{Bannai:2017gv}]\label{lem:longestIsRightnode}
  Let $w$ be a Lyndon word. For any interval $[i..j]$ except for
  $[1..|w|]$, $[i..j]$ corresponds to a right node of the Lyndon tree
  if and only if $w[i..j]$ is the longest Lyndon word that starts at $i$.
\end{lemma}
Thus,
as before, we have that
$\rho = x_{\rho} \lambda_{\rho} y_{\rho} = x'_{\rho} \lambda_{\rho}^k y_{\rho}$,
where
$\lambda_{\rho}^k = r_{\rho}$ is the L-root interval of $\rho$,
$x'_{\rho}$ is a possibly empty proper suffix of $\lambda_{\rho}$,
and that each  occurrence of $\lambda_{\rho}$ corresponds to a right node
in one the Lyndon trees.
Now, $|x_{\rho}| = |x'_{\rho}| + (k-1)|\lambda_{\rho}|$,
and we distribute this sum among each of the $k$ occurrences of the L-root
as follows: $|x'_{\rho}|$ for the leftmost occurrence
(i.e., the periodicity only extends $|x'_{\rho}|$ symbols to the left of the occurrence),
or $|\lambda_{\rho}|$ otherwise (i.e., the periodicity extends at least $|\lambda_{\rho}|$ symbols to the left of the occurrence).

Next, consider how long the periodicity can extend to the left of
each occurrence of $\lambda_\rho$ by looking at the Lyndon tree.
Since $\lambda_\rho$ corresponds to a right node,
$w_{\rho} = z_{\rho} \lambda_{\rho}$ for some
Lyndon words $w_{\rho}$ and $z_{\rho}$.
When $|z_{\rho}| \leq |\lambda_\rho|$,
$z_{\rho}$ cannot be a suffix of $\lambda_{\rho}$, since that would imply that
$w_{\rho} = z_{\rho} \lambda_{\rho} < \lambda_{\rho} < z_{\rho}$, a contradiction.
Thus, for the occurrence of L-root $\lambda_{\rho}$ in $w_{\rho}$,
the periodicity can extend at most
$|z_{\rho}|$ symbols (more precisely, $|z_{\rho}|-1$ symbols).

Let $\mathcal{S}(n)$ denote the maximum of the total sum of all $|x_{\rho}|$ for
all potential L-roots $\lambda_{\rho}$ that correspond to a right node in a
(single) Lyndon tree for any string of length $n$.
From the above arguments, we have
$\mathcal{S}(n) = 0$ if $n = 1$,
and otherwise,
$\mathcal{S}(n) \leq  \max \{\mathcal{S}(n_1) + \mathcal{S}(n_2) + \min\{ n_1, n_2 \} \mid n_1,n_2 > 0 \mbox{ and } n_1+n_2 = n \}.$
We can show by induction
that $\mathcal{S}(n)$ can be bounded by $\frac{n}{2}\log_2 n$.
\begin{lemma} \label{lem:Cnlqnlogn}
  $\mathcal{S}(n) \le \frac{n}{2}\log_2 n$.
\end{lemma}
\begin{proof}
Clearly, when $|n|$ = 1, $0 = \mathcal{S}(n) \leq \frac{1}{2}\log_2 1 = 0$.
  For $n \geq 2$, assume that the lemma holds for any value less than $n$.
  Then,
  \begin{eqnarray*}
    \mathcal{S}(n) &\leq&\max \left\{\mathcal{S}(n_1) + \mathcal{S}(n_2) + \min\{ n_1, n_2 \} \mid n_1,n_2 \neq 0 \mbox{ and } n_1+n_2 = n \right\}\\
    &\leq& \max \left\{ \frac{n - k_n}{2}\log_2 (n - k_n) + \frac{k_n}{2}\log_2 k_n + k_n ~\middle|~ 1 \le k_n \leq \frac{n}{2}\right\}\\
    &=& \frac{1}{2} \max \left\{  ((n - k_n)\log_2 (n - k_n) + k_n\log_2 k_n + 2k_n)~\middle|~ 1 \le k_n \leq \frac{n}{2}\right\}\\
    &\le& \frac{1}{2} \left(\left(n - \frac{n}{2}\right)\log_2 \left(n - \frac{n}{2}\right) + \frac{n}{2}\log_2 \frac{n}{2} + n\right)\\
&=&   \frac{1}{2} \left(n\log_2 \frac{n}{2} + n\right)  = \frac{n}{2} \log_2 n.
  \end{eqnarray*}
  The third inequality follows since the second derivative of the above function is positive
  and thus the function is maximized when $k_n =n/2$.
\end{proof}
Now, since any occurrence of an L-root corresponds to a right
node in one of the two Lyndon trees,
we have
\[\PS(n) \leq \max_{w\in\Sigma^n}\sum_{\rho\in\Runs(w)}|x_\rho| \leq 2\cdot \mathcal{S}(n) \leq n\log_2 n.\]

\section*{Acknowledgments}
We would like to thank the anonymous reviewers for pointing out and correcting errors in the submitted version of the paper.

This work was supported by JSPS KAKENHI Grant Numbers
JP20H04141 (HB), JP20J11983 (TM), and JP18K18002 (YN).


\begin{thebibliography}{10}
\providecommand{\url}[1]{\texttt{#1}}
\providecommand{\urlprefix}{URL }
\providecommand{\doi}[1]{https://doi.org/#1}

\bibitem{DBLP:conf/soda/BannaiIINTT15}
Bannai, H., I, T., Inenaga, S., Nakashima, Y., Takeda, M., Tsuruta, K.: A new
  characterization of maximal repetitions by {L}yndon trees. In: Indyk, P.
  (ed.) Proceedings of the Twenty-Sixth Annual {ACM-SIAM} Symposium on Discrete
  Algorithms, {SODA} 2015, San Diego, CA, USA, January 4-6, 2015. pp. 562--571.
  {SIAM} (2015). \doi{10.1137/1.9781611973730.38}

\bibitem{Bannai:2017gv}
Bannai, H., I, T., Inenaga, S., Nakashima, Y., Takeda, M., Tsuruta, K.: The
  “runs” theorem. SIAM Journal on Computing  \textbf{46}(5),  1501--1514
  (2017). \doi{10.1137/15M1011032}

\bibitem{BARCELO199093}
Barcelo, H.: On the action of the symmetric group on the free {L}ie algebra and
  the partition lattice. Journal of Combinatorial Theory, Series A
  \textbf{55}(1),  93--129 (1990). \doi{10.1016/0097-3165(90)90050-7}

\bibitem{Berstel:2007fb}
Berstel, J., Perrin, D.: The origins of combinatorics on words. Eur. J. Comb.
  \textbf{28}(3),  996–--1022 (Apr 2007). \doi{10.1016/j.ejc.2005.07.019}

\bibitem{10.2307_1970044}
Chen, K.T., Fox, R.H., Lyndon, R.C.: Free differential calculus, iv. the
  quotient groups of the lower central series. Annals of Mathematics
  \textbf{68}(1),  81--95 (1958). \doi{10.2307/1970044}

\bibitem{CROCHEMORE2008796}
Crochemore, M., Ilie, L.: Maximal repetitions in strings. Journal of Computer
  and System Sciences  \textbf{74}(5),  796--807 (2008).
  \doi{10.1016/j.jcss.2007.09.003}

\bibitem{10.1007/978-3-319-46049-9_3}
Crochemore, M., Iliopoulos, C.S., Kociumaka, T., Kundu, R., Pissis, S.P.,
  Radoszewski, J., Rytter, W., Walen, T.: Near-optimal computation of runs over
  general alphabet via non-crossing {LCE} queries. In: Inenaga, S., Sadakane,
  K., Sakai, T. (eds.) String Processing and Information Retrieval - 23rd
  International Symposium, {SPIRE} 2016, Beppu, Japan, October 18-20, 2016,
  Proceedings. Lecture Notes in Computer Science, vol.~9954, pp. 22--34 (2016).
  \doi{10.1007/978-3-319-46049-9\_3}

\bibitem{DBLP:journals/tcs/CrochemoreIKRRW14}
Crochemore, M., Iliopoulos, C.S., Kubica, M., Radoszewski, J., Rytter, W.,
  Walen, T.: Extracting powers and periods in a word from its runs structure.
  Theor. Comput. Sci.  \textbf{521},  29--41 (2014).
  \doi{10.1016/j.tcs.2013.11.018}

\bibitem{crochemore1995squares}
Crochemore, M., Rytter, W.: Squares, cubes, and time-space efficient string
  searching. Algorithmica  \textbf{13}(5),  405--425 (1995).
  \doi{10.1007/BF01190846}

\bibitem{Deza:2015ih}
Deza, A., Franek, F., Thierry, A.: How many double squares can a string
  contain? Discrete Applied Mathematics  \textbf{180},  52--69 (2015).
  \doi{10.1016/j.dam.2014.08.016}

\bibitem{periodicity_lemma}
Fine, N.J., Wilf, H.S.: Uniqueness theorems for periodic functions. Proceedings
  of American Mathematical Society  \textbf{16}(1),  109--114 (1965).
  \doi{10.1090/S0002-9939-1965-0174934-9}

\bibitem{DBLP:conf/spire/0001HIL15}
Fischer, J., Holub, S., I, T., Lewenstein, M.: Beyond the runs theorem. In:
  Iliopoulos, C.S., Puglisi, S.J., Yilmaz, E. (eds.) String Processing and
  Information Retrieval - 22nd International Symposium, {SPIRE} 2015, London,
  UK, September 1-4, 2015, Proceedings. Lecture Notes in Computer Science,
  vol.~9309, pp. 277--286. Springer (2015). \doi{10.1007/978-3-319-23826-5\_27}

\bibitem{Fraenkel:1998wi}
Fraenkel, A.S., Simpson, J.: How many squares can a string contain? Journal of
  Combinatorial Theory, Series A  \textbf{82}(1),  112--120 (1998).
  \doi{10.1006/jcta.1997.2843}

\bibitem{fraenkel1999exact}
Fraenkel, A.S., Simpson, J.: The exact number of squares in {Fibonacci} words.
  Theoretical Computer Science  \textbf{218}(1),  95--106 (1999).
  \doi{10.1016/S0304-3975(98)00252-7}

\bibitem{DBLP:conf/cpm/GawrychowskiKRW16}
Gawrychowski, P., Kociumaka, T., Rytter, W., Walen, T.: Faster longest common
  extension queries in strings over general alphabets. In: Grossi, R.,
  Lewenstein, M. (eds.) 27th Annual Symposium on Combinatorial Pattern
  Matching, {CPM} 2016, June 27-29, 2016, Tel Aviv, Israel. LIPIcs, vol.~54,
  pp. 5:1--5:13. Schloss Dagstuhl - Leibniz-Zentrum f{\"{u}}r Informatik
  (2016). \doi{10.4230/LIPIcs.CPM.2016.5}

\bibitem{Ilie:2007hv}
Ilie, L.: A note on the number of squares in a word. Theoretical Computer
  Science  \textbf{380}(3),  373--376 (2007). \doi{10.1016/j.tcs.2007.03.025},
  combinatorics on Words

\bibitem{DBLP:conf/stoc/KempaK19}
Kempa, D., Kociumaka, T.: String synchronizing sets: sublinear-time {BWT}
  construction and optimal {LCE} data structure. In: Charikar, M., Cohen, E.
  (eds.) Proceedings of the 51st Annual {ACM} {SIGACT} Symposium on Theory of
  Computing, {STOC} 2019, Phoenix, AZ, USA, June 23-26, 2019. pp. 756--767.
  {ACM} (2019). \doi{10.1145/3313276.3316368}

\bibitem{DBLP:journals/siamcomp/KnuthMP77}
Knuth, D.E., Jr., J.H.M., Pratt, V.R.: Fast pattern matching in strings. {SIAM}
  J. Comput.  \textbf{6}(2),  323--350 (1977). \doi{10.1137/0206024}

\bibitem{814634}
Kolpakov, R.M., Kucherov, G.: Finding maximal repetitions in a word in linear
  time. In: 40th Annual Symposium on Foundations of Computer Science, {FOCS}
  '99, 17-18 October, 1999, New York, NY, {USA}. pp. 596--604. {IEEE} Computer
  Society (1999). \doi{10.1109/SFFCS.1999.814634}

\bibitem{KOSOLOBOV2016241}
Kosolobov, D.: Computing runs on a general alphabet. Information Processing
  Letters  \textbf{116}(3),  241--244 (2016).
  \doi{https://doi.org/10.1016/j.ipl.2015.11.016}

\bibitem{lothaireCombinatoricsOnWords}
Lothaire, M.: Combinatorics on Words. Addison-Wesley, Reading, MA (1983)

\bibitem{Lyndon:1954it}
Lyndon, R.C.: On {Burnside's} problem. Transactions of the American
  Mathematical Society  \textbf{77}(2),  202--202 (Feb 1954).
  \doi{10.2307/1990868}

\bibitem{PUGLISI2008165}
Puglisi, S.J., Simpson, J., Smyth, W.: How many runs can a string contain?
  Theoretical Computer Science  \textbf{401}(1),  165--171 (2008).
  \doi{10.1016/j.tcs.2008.04.020}

\bibitem{Rytter2006STACS}
Rytter, W.: The number of runs in a string: Improved analysis of the linear
  upper bound. In: Durand, B., Thomas, W. (eds.) {STACS} 2006, 23rd Annual
  Symposium on Theoretical Aspects of Computer Science, Marseille, France,
  February 23-25, 2006, Proceedings. Lecture Notes in Computer Science,
  vol.~3884, pp. 184--195. Springer (2006). \doi{10.1007/11672142\_14}

\end{thebibliography}
\end{document}